\pdfoutput =1

\documentclass[pra,onecolumn,superscriptaddress]{article}

\usepackage[T1]{fontenc}
\usepackage[utf8]{inputenc}
\usepackage{amsmath,amsfonts,amsthm,amssymb}
\usepackage{subeqnarray,array}
\usepackage{setspace}
\usepackage{graphics,graphicx,color}
\usepackage{url,hyperref}
\usepackage{enumerate}
\usepackage{float} 
\usepackage{multirow}
\usepackage{hhline}
\usepackage[margin=1.2in]{geometry}
\usepackage{braket}
\usepackage{dsfont} 
\usepackage{authblk}
\usepackage[makeroom]{cancel}
\usepackage{diagbox}
\usepackage{hhline}
\usepackage[math]{cellspace}
\newtheorem{theorem}{Theorem}

\newtheorem{definition}[theorem]{Definition}
\newtheorem{lem}[theorem]{Lemma}

\usepackage[scr=boondoxo,scrscaled=1.05]{mathalfa}

\newcommand{\expectation}[2]{\mbox{$ \Braket{ #1 | #2 | #1 } $}}

\newcommand{\proj}[1]{ \ket{#1}\!\bra{#1} }

\newcommand{\beq}{\begin{eqnarray}}
\newcommand{\eeq}{\end{eqnarray}}
\newcommand\abs[1]{\left|#1\right|}
\newcommand\norm[1]{\left\Vert#1\right\Vert}

\newcommand{\ps}{\ket{\psi}}

\newcommand{\C}{\mathbb C}
\newcommand{\N}{\mathbb N}

\newcommand{\veven}[1]{V_{#1}^\text{even}}
\newcommand{\vodd}[1]{V_{#1}^\text{odd}}

\newcommand{\saz}{\sigma_\alpha^z}
\newcommand{\sax}{\sigma_\alpha^x}

\newcommand{\nin}{\not\in}
\newcommand{\seq}{\subseteq}
\newcommand{\m}[1]{\mathcal{#1}}
\renewcommand{\set}[1]{\left\{#1\right\}}

\newcommand{\induce}{attain}

\DeclareMathOperator{\Schmidt}{Sch}
\DeclareMathOperator{\tr}{Tr}

\newcommand{\matwo}[4]{\left( \begin{array}{cc}#1& #2 \\ #3 &#4 \end{array} \right)}

\begin{document}
\title{Unconditional separation of finite and infinite-dimensional quantum correlations}

\author{Andrea Coladangelo}
\author{Jalex Stark}

\affil{Computing and Mathematical Sciences,	
Caltech\\ 	
{\{acoladan,jalex\}@caltech.edu}
}

\maketitle

\abstract{
\noindent Determining the relationship between quantum correlation sets is a long-standing open problem. 
The most well-studied part of the hierarchy is captured by the chain of inclusions
$\m C_q \seq \m C_{qs} \subsetneq \m C_{qa} \seq \m C_{qc}$. 
The separation $\m C_{qs} \neq \m C_{qa}$, showing that the set of quantum spatial correlations is not closed, was proven in breakthrough work by Slofstra \cite{slofstra2016tsirelson,slofstra2017set}.
Resolving the question of $\m C_{qa} = \m C_{qc}$ would resolve the Connes Embedding Conjecture and would represent major progress in the mathematical field of operator algebras.
In this work, we resolve the ambiguity in the first inclusion, showing that $\m{C}_q \neq \m{C}_{qs}$. We provide an explicit construction of a correlation that can be attained on a tensor product of infinite-dimensional Hilbert spaces but not finite-dimensional ones. This property is also conjectured to be possessed by any correlation which maximally violates the $I_{3322}$ inequality.

\section{Introduction}

Let Alice and Bob be two spacelike-separated provers. Consider the scenario in which a verifier sends one question to each prover and receives an answer from each prover. The behaviour of the provers is captured by the joint distribution of their answers as a function of their questions. We refer to this data as a \emph{bipartite correlation}. One of the most fundamental questions one can ask about a correlation is ``in which models of physics can the correlation be realized?''. Some correlations can be realized in classical physics if one allows the provers to share randomness ahead of time. However, it is well-known that some correlations require quantum resources to realize \cite{Bell1964}. In fact, different models of quantum mechanics admit different sets of correlations. Characterizing the relationship between these sets is a long-standing problem. 

We say that a correlation is in the set of \emph{quantum correlations} $\m C_q$ if there is a finite-dimensional state $\ket\psi$ and finite-dimensional projective measurements $\set{A_x^a}, \set{B_y^b}$ so that 
\begin{equation}\label{eq:definition-of-induce}
	p(a,b| x,y) = \braket{\psi|A_x^a \otimes B_y^b|\psi},
\end{equation}
where $p(a,b| x,y)$ is the probability that Alice answers $a$ and Bob answers $b$, given that Alice was asked question $x$ and Bob was asked question $y$. We say that a correlation is in the set of \emph{quantum spatial correlations} $\m C_{qs}$ if Equation \eqref{eq:definition-of-induce} holds with a state and measurements that are possibly infinite-dimensional. Notice that $\m C_q \subseteq \m C_{qs}$. 

We say that correlation is in the set of \emph{quantum-approximate correlations} $\m C_{qa}$ if it is arbitrarily well-approximated by correlations in $\m C_{q}$. In other words, $\m C_{qa}$ is the closure of $\m C_{q}$. From \cite{scholz2008tsirelson}, we know that also $\m C_{qs} \subseteq \m C_{qa}$, hence $\m C_{qa}$ is also the closure of $\m C_{qs}$. Finally, one can define the set of \emph{quantum commuting correlations} $\m C_{qc}$ as the set of possibly infinite-dimensional quantum correlations arising in the commuting operator model. Note that these definitions all assume finite question and answer sets.

Following breakthrough work by Slofstra, showing first that $\m C_{qs} \neq \m C_{qc}$ \cite{slofstra2016tsirelson} and later strengthening this to $\m C_{qs} \neq \m C_{qa}$ \cite{slofstra2017set}, the known hierarchy so far is the following:
\begin{equation}\label{eq:four-correlation-sets}
	\m C_q \subseteq \m C_{qs} \subsetneq \m C_{qa} \subseteq \m C_{qc}.
\end{equation}

This ``four correlation sets'' picture, along with the explicit study of $\m C_{qs}$, was introduced by Paulsen and coauthors. \cite{paulsen2014problems,paulsen2015quantum,dykema2016synchronous} Our main theorem is that the first inclusion in Equation \eqref{eq:four-correlation-sets} is also strict: $\m C_q \neq \m C_{qs}$. In particular, we give an explicit correlation which can be attained in infinite dimensions, and we show that it cannot be attained in finite dimensions.

\subsection{Related work}
The primary related work is the study of the $I_{3322}$ inequality, first introduced in \cite{I3322Froissart}. In \cite{I3322Pal}, P\'al and V\'ertesi give numerical evidence suggesting that finite-dimensional states are not enough to attain maximal violation of the inequality, but they conjecture that infinite-dimensional states suffice. In this paper, we are providing a correlation on slightly larger question and answer sets, which we prove exhibits precisely this conjectured property. 

In a related line of work, \cite{slofstra2017set}, and the subsequent \cite{dykema2017non, ji2018three}, provide non-local games which require arbitrarily high-dimensional strategies to attain win probabilities arbitratily close to $1$. However, the sequences of ideal strategies for these games do not have a limit, since they require arbitrarily high amounts of entanglement entropy. So they separate $\m C_{qs}$ from $\m C_{qa}$ but do not shed any light on the relationship between $\m C_{q}$ and $\m C_{qs}$.

Other related works have considered the relaxed setting which allows either the question set or the answer set to be countably infinite. To the best of the authors' knowledge, \cite{manvcinska2014unbounded} was the first result giving a non-local game (with classical questions) whose optimal winning probability can be approximated arbitrarily well, but not achieved perfectly, with finite-dimensional states. The game has two questions per party and countably infinite answer sets. However, again the sequence of correlations presented there does not have a limit, since the correlations consist of uniform distributions on increasingly large sets. On the other hand, a previous work by the present authors \cite{cs2017separation} provides sequences of correlations, on either infinite question sets (and finite answer sets) or infinite answer sets (and finite question sets), that cannot be attained in finite-dimensions but can be attained explicitly in infinite dimensions. The present result is a significant strengthening of this, achieving the same property while overcoming the major hurdle of maintaining both finite question and answer sets.

\subsection{Our result}
Inspired by the self-testing techniques of \cite{CGS16} and \cite{cs2017separation}, we construct an explicit correlation, on question sets of sizes $4$ and $5$ and answer sets of size $3$ for both parties, which can be attained using an infinite-dimensional quantum strategy, but not any finite-dimensional one.

So far, we have referred to $\m C_q$ as the set of quantum correlations and and $\m C_{qs}$ as the set of quantum spatial correlations in order to match the literature. Since we will be using only these two of the four correlation sets, we will find it convenient to now refer to them as the set of \emph{finite-dimensional quantum correlations} and the set of \emph{infinite-dimensional quantum correlations}. In particular, let $\m{C}_q^{m,n,r,s}$ ($\m{C}_{qs}^{m,n,r,s}$) be the set of finite-dimensional (resp. infinite-dimensional) quantum correlations on question sets of sizes $m$ and $n$ and answer sets of sizes $r$ and $s$. (Notice that $\m C_q = \bigcup_{m,n,r,s<\infty}\m C_q^{m,n,r,s}$ and similarly for $\m C_{qs}$.) Our result is the following:

\begin{theorem}
\label{thm: main}
    $\m{C}_q^{4,5,3,3} \neq \m{C}_{qs}^{4,5,3,3} $.
\end{theorem}


\subsection{A cartoon proof}
\label{sec: cartoon proof}
We give a very concise overview of the structure of our proof of Theorem \ref{thm: main}. To explain the argument, we start by giving an idealized version that runs against a barrier, and then talk about how to avoid the barrier.

We introduce an ideal correlation $p^*$, and then use self-testing techniques to guarantee that, given any strategy on some bipartite state $\ket{\psi}$ that \induce s $p^*$, there exist local isometries which take the state $\ket\psi$ to states of a certain form. In fact, suppose that ``by magic'' we knew that achieving the ideal correlation guaranteed the following two things. First, there is a local isometry $\Phi = \Phi_A\otimes \Phi_B$ such that 
\begin{equation}\label{eq:cartoon-decomposition-1}
    \Phi(\ket\psi) = \frac{1}{\sqrt{1+\alpha^2}}(\ket{00} + \alpha\ket{11})\otimes \ket{aux}.
\end{equation}
 
Next, there is another local isometry $\Phi'$ such that
\begin{equation}\label{eq:cartoon-decomposition-2}
    \Phi'(\ket\psi)=  \ket{\phi} \oplus \frac{1}{\sqrt{1+\alpha^2}}(\ket{00}+\alpha\ket{11})\otimes \ket{aux'},
\end{equation}
where $\oplus$ denotes a direct sum and the state $\ket{\phi}$ is separable, i.e.\ has Schmidt rank $1$. Then suppose towards a contradiction that $\ket\psi$ were finite-dimensional. From the first condition we see that the Schmidt rank of the state is even, while from the second condition we see that the Schmidt rank of the state is odd; contradiction.

In the above, the ``magic'' happens where we assume that $\ket\phi$ is separable. In general, any correlation that is \induce d using a separable $\ket\phi$ could also be \induce d by tensoring with extra entanglement and not making use of it in the measurements, so we will not be able to assume that $\ket\phi$ is separable. However, our main argument will still decompose $\ket\psi$ into two ways as in equations \eqref{eq:cartoon-decomposition-1} and \eqref{eq:cartoon-decomposition-2}. In place of the odd / even constraints, we will show that these decompositions partition the Schmidt coefficients into two different ways so that the set of nonzero Schmidt coefficients of $\ket{\psi}$ is in bijection with a proper subset of itself.


\subsection{Organization}
Section \ref{sec: preliminaries} covers some preliminary notions. Section \ref{sec: direct sums} formalizes the notion of a direct sum of correlations and proves that a certain block structure in a correlation implies a similar direct sum decomposition of the state and measurements achieving the correlation. 
In Section \ref{sec: separating correlation}, we describe the separating correlation by specifying the infinite dimensional state and measurements that \induce\ it exactly. In Section \ref{sec: proj}, we apply self-testing techniques to establish properties of any state and measurements achieving the separating correlation; these properties will be similar to Equations \eqref{eq:cartoon-decomposition-1} and \eqref{eq:cartoon-decomposition-2}. Finally in Section \ref{sec: schmidt}, we will use these properties of the state to show that it has infinitely many nonzero Schmidt coefficients.

\section{Preliminaries}
\label{sec: preliminaries}
For a positive integer $n$, we denote by $[n]$ the set $\{1,..,n\}$. $\delta_{ij}$ is the Kronecker delta. For a Hilbert space $\mathcal{H}$, $\mathcal{L}(\m H)$ is the space of linear operators on $\mathcal{H}$. 
Define the Pauli matrices
\begin{equation*}
	\sigma^z = \matwo 100{-1}\text{, and }\sigma^x = \matwo 0110.
\end{equation*}
For an operator $T \in \mathcal{L}(\m H)$ and a subspace $\m H' \seq \m H$ invariant under $T$, we denote by $T|_{\m H'} \in \mathcal{L}(\m H')$ the restriction of $T$ to $\m H'$. Let $\C^\N$ denote the Hilbert space of square-summable sequences, sometimes called $\ell^2(\C)$. We endow it with a standard basis $\set{\ket i: i\in \N}$. Formally, $\C^\N = \set{\sum_i a_i \ket i: \sum_i \abs{a_i^2} < \infty}$.

\renewcommand{\induce}{induce}
\paragraph{Correlations and quantum strategies}
Given sets $\m{X}$,$\m{Y}$, $\m{A}$, $\m{B}$, a (bipartite) \textit{correlation} is a collection $\{p(a,b|x,y): a\in \m{A}, b \in \m{B}\}_{(x,y)\in \m{X}\times \m{Y}}$, where each $p(\cdot,\cdot|x,y)$ is a probability distibution over $\m A\times \m B$. We interpret the correlation as describing the outcomes of a measurement scenario with two parties, say Alice and Bob. $p(a,b|x,y)$ is the probability that Alice outputs $a$ and Bob outputs $b$, given that Alice used measurement setting $x$ and Bob used setting $y$. $\m{X}$ and $\m{Y}$ are referred to as the \emph{question sets}, while $\m{A}$ and $\m{B}$ are referred to as as the \emph{answer sets}.

Given question sets and answer sets $\m{X}$, $\m{Y}$, $\m{A}$, $\m{B}$,
a \textit{quantum strategy} is specified by Hilbert spaces $\m{H}_A$ and  $\m{H}_B$, a pure state $\ket{\psi} \in \m{H}_A \otimes \m{H}_B$, and projective measurements $\{\Pi^a_{A_x}\}_a$ on $\m{H}_A$, $\{\Pi^b_{B_y}\}_b$ on $\m{H}_B$, for $x \in \m{X}, y \in \m{Y}$. We say that it \emph{induces correlation $p$} if
\begin{equation}
	p(a,b|x,y) = \bra{\psi} \Pi_{A_x}^a \otimes \Pi_{B_y}^b \ket{\psi} \text{ for all }a \in \m{A}, b \in \m{B}, x \in \m{X}, y \in \m{Y}.
\end{equation}
Sometimes we refer to a quantum strategy as a triple $\left(\ket{\psi}, \{\Pi^a_{A_x}\}_a, \{\Pi^b_{B_y}\}_b \right)$. If we wish to emphasize the underlying Hilbert space, we write
$\left(\ket{\psi} \in \m{H}_A \otimes \m{H}_B, \{\Pi^a_{A_x}\}_a, \{\Pi^b_{B_y}\}_b \right)$.
Notice that we have chosen our state to be pure and our measurements to be projective. This choice is without loss of generality. The most general measurements are modeled by POVMs, but Naimark's dilation theorem implies that any correlation \induce d using POVMs can also be \induce d using projective measurements (possibly of larger dimension). Likewise, any correlation \induce d by a mixed state can also be \induce d by using a purification of that state.
 We sometimes describe a quantum strategy by specifying an observable for each question. The observables in turn specify the projectors through their eigenspaces. 

A correlation is said to be quantum if there exists a quantum strategy that \induce s it. We refine this, and we say that a quantum correlation is \textit{finite-dimensional} (\textit{infinite-dimensional}) if it is \induce d by a quantum strategy on a finite-dimensional (infinite-dimensional) Hilbert space. We denote by $\m{C}_q^{m,n,r,s}$ and $\m{C}_{qs}^{m,n,r,s}$ respectively the sets of finite and infinite-dimensional quantum correlations on question sets of sizes $m, n$ and answer sets of sizes $r,s$. 


\paragraph{Self-testing} As we will be borrowing inspiration and techniques from the field of device-independent self-testing, we provide a formal definition: 

\begin{definition}[Self-testing]
We say that a correlation $\{p^*(a,b|x,y): a \in \m{A}, b \in \m{B}\}_{x \in \m{X}, y \in \m{Y}}$ self-tests a strategy $\left(\ket{\Psi}, \{\tilde{\Pi}_{A_x}^{a}\}_{a}, \{\tilde{\Pi}_{B_y}^{b}\}_{b} \right)$ if, for any strategy $\left(\ket{\psi}, \{\Pi_{A_x}^a\}_a, \{\Pi_{B_y}^b\}_b \right)$ that \induce s $p^*$, there exists a local isometry $\Phi= \Phi_A\otimes\Phi_B$ and an auxiliary state $\ket{aux}$ such that, for all $x \in \m X,y \in \m Y,a \in \m A,b\in \m B$,
\begin{align}
\Phi(\ket{\psi}) &=  \ket{\Psi} \otimes \ket{aux} \label{eq: state}\\ \label{eq: measurements}
\Phi(\Pi_{A_x}^{a}\otimes \Pi_{B_y}^{b}\ket{\psi}) &=  \tilde{\Pi}_{A_x}^{a}\otimes \tilde{\Pi}_{B_y}^{b} \ket{\Psi} \otimes \ket{aux} 
\end{align}
\end{definition}
Sometimes, we refer to \textit{self-testing of the state} when we are only concerned with the guarantee of equation \eqref{eq: state}, and not \eqref{eq: measurements}. 

\paragraph{Tilted CHSH}

We introduce the tilted CHSH inequality \cite{Acin12}, which is a building block for the correlation that appears in this work. First, we recall the CHSH inequality. It states that for binary observables $A_0, A_1$ on Hilbert space $\m H_A$ and binary observables $B_0, B_1$ on Hilbert space $\m H_B$ together with a product state $\ket\phi = \ket{\phi_A} \otimes \ket{\phi_B}$, we have
\begin{equation}
	\expectation \phi {A_0B_0 + A_0B_1 +A_1B_0 - A_1B_1} \leq 2,
\end{equation}
where the maximum is achieved (for example setting all observables to identity). However, if instead of the product state $\ket\phi$ we allow an entangled state $\ket\psi$, then the right-hand side of the inequality increases to $2\sqrt2$. This maximum requires a maximally entangled state to achieve. In this work, we would like to use an inequality that requires a non-maximally entangled state to achieve the maximum; this is the tilted CHSH inequality. Given a real parameter\footnote{This parameter is usually called $\alpha$ in the literature, since the Bell inequality is the fundamental object of study. For our purposes, the object of study is the state and measurements appearing in Definition \ref{def: ideal tilted chsh}, so we reserve the letter $\alpha$ for a parameter appearing there. These $\beta$ and $\alpha$ are related by an invertible function.} $\beta \in [0,2]$, for a product state $\ket{\phi} = \ket{\phi_A} \otimes \ket{\phi_B}$,
\begin{equation}
	\expectation \phi {\beta A_0 + A_0B_0 + A_0B_1 +A_1B_0 - A_1B_1} \leq 2+\beta.
\end{equation}
For entangled $\ket\psi$, we have instead that
\begin{equation}\label{eq:tiltedchsh}
	\expectation \psi {\beta A_0 + A_0B_0 + A_0B_1 +A_1B_0 - A_1B_1} \leq \sqrt{8+2\beta^2}.
\end{equation}
The maximum in the tilted CHSH inequality is attained by the following strategy:
\begin{definition}[Ideal strategy for tilted CHSH]
\label{def: ideal tilted chsh}
Given parameter $\beta$, let $\sin 2\theta = \sqrt{\frac{4-\beta^2}{4+\beta^2}}$, $\mu = \arctan \sin 2\theta$, and $\alpha = \tan \theta$. Define the \emph{$\alpha$-tilted Pauli operators} as
\begin{equation}
 	\saz := \cos \mu \sigma^z + \sin \mu \sigma^x\text{, and } \sax := \cos \mu \sigma^z - \sin \mu \sigma^x.
 \end{equation} 
The ideal strategy for tilted CHSH with parameter $\beta$ (i.e.\ achieving maximal violation of \eqref{eq:tiltedchsh}) consists of the joint state $\ket{\Psi} = \cos \theta (\ket{00} + \alpha \ket{11})$ and observables $A_0, A_1$ and $B_0, B_1$ with $A_0 = \sigma^z$, $A_1 = \sigma^x$, $B_0 = \sigma_\alpha^z$ and $B_1 = \sigma_\alpha^x$. For each observable, we associate the projection onto the $+1$-eigenspace with answer $0$ and the projection onto the $-1$-eigenspace with answer $1$.

\end{definition}

Since in the present work we are primarily concerned with the ratio of the coefficients of the ideal state, we refer to the correlation defined by the ideal strategy of Definition \ref{def: ideal tilted chsh} as the \textit{ideal tilted CHSH correlation for ratio $\alpha$}. In the remainder of the paper, we use the correlation along with the ideal strategy, but we will forget the Bell inequality \eqref{eq:tiltedchsh} that motivates them. In particular, we will use the following lemma.



\begin{lem}[\cite{Bamps15}]
\label{Bamps lemma}

The tilted CHSH correlation for ratio $\alpha$ self-tests the strategy of Definition \ref{def: ideal tilted chsh}.
\end{lem}

\paragraph{Correlation tables}
A convenient way to describe correlations is through \textit{correlation tables}. A correlation $p$ on $\m{X}$, $\m{Y}$, $\m{A}$, $\m{B}$ is completely specified by correlation tables $T_{xy}$ for $x \in \m{X}, y \in \m{Y}$, with entries $T_{xy}(a,b) = p(a,b |x,y)$. See Table \ref{tab:txy0}.
\begin{table}[h]
\label{tab:txy0}
\caption{The correlation table on question $(x,y)$ of a correlation on answer sets $\m{A} = \m{B} = \{0,1\}$.}
\begin{center}
\begin{tabular}{| c || Sc | Sc |}
	\hline
	\diagbox[width=2.5em]{$a$}{$b$} & 0 & 1 \\ \hhline{|===|}
	0 & $ p(0,0|x,y)$ & $p(0,1|x,y)$ \\ 
	\hline
	1 & $p(1,0|x,y)$ &  $p(1,1|x,y)$ \\ 
	\hline
\end{tabular}
\end{center}
\end{table}

As mentioned earlier, we will make use of the ideal tilted CHSH correlation as a building block for our separating correlation. For $x,y \in \{0,1\}$ and $\alpha \in (0,1)$, we denote by $\mbox{CHSH}_{x,y}^{\alpha}$ the correlation table on question $x,y$ for the ideal tilted CHSH correlation for ratio $\alpha$. 

For $\omega \in [0,1]$ and a correlation table $T_{xy}$, we write $\omega \cdot T_{xy}$ to denote entry-wise multiplication of $T_{xy}$ by $\omega$. We may refer to $\omega$ as a \emph{weight}. 

\section{Direct sums of correlations}
\label{sec: direct sums}
In this section, we introduce the notion of a direct sum of correlations. We will later use this to build our desired correlation out of tilted CHSH building blocks. Lemma \ref{lem: key} will allow us to characterize the strategies for the desired correlation from self-testing results about its direct summands. In particular, these strategies also decompose, in a sense made precise below, as a direct sum of strategies corresponding to the direct summands.
The proof is somewhat technical, and the ideas in the proof are not necessary to understand the rest of the paper. Some of the ideas in the proof have appeared ad-hoc in previous works on constructing quantum correlations block-by-block \cite{CGS16}, \cite{coladangelo2018chsh}. We package these arguments into a lemma since it may be of independent interest.
First, we define formally a direct sum of correlations.
\begin{definition}[Direct sum of correlations]
Let $p$ be a correlation on $\m{X}, \m{Y}, \m{A}$, $\m{B}$. Suppose for some positive integer $l$, for $i \in [l]$, there exist partitions
$\m{A} = \bigsqcup_{i=1}^{l}\m{A}_i$, $\m{B} = \bigsqcup_{i=1}^{l}\m{B}_i$, real numbers $\omega_i \geq 0$ with $\sum_{i=1}^l \omega_i =1$, and correlations $p_i$ on $\m{X},\m{Y}, \m{A}_i, \m{B}_i$ such that for all $i,j \in [l],$ 
$a\in \m A_i, b \in \m B_j, x\in \m X, y\in \m Y$,
\begin{equation}\label{eq:direct-sum-condition}
	p(a,b|x,y) = \delta_{ij} \omega_i p_i(a,b|x,y).
\end{equation}
Then we say that $p$ is a direct sum of the $p_i$, and we write $p = \oplus_{i=1}^l \omega_i p_i$. We sometimes refer to the $p_i$ as \emph{blocks} of $p$ and the $\omega_i$ as \emph{weights} of the blocks. We give a visual interpretation of condition \eqref{eq:direct-sum-condition} in Table \ref{table:direct-sum}.
\end{definition}

\begin{table}[H]
\caption{The correlation table for $p = \oplus_i \omega_ip_i$ on questions $x,y$. $T_{xy}^{(i)}$ is the correlation table for correlation $p_i$ on questions $x,y$.}
\label{table:direct-sum}
\begin{center}
\begin{tabular}{| c || Sc | Sc | Sc | Sc | Sc | Sc | Sc |}
	\hline
	\diagbox[width=2.5em]{$a$}{$b$} & \multicolumn{2}{c|}{$\m B_1$}  & $\cdots$ & \multicolumn{2}{c|}{$\m B_l$}		\\ \hhline{|======|} 
	\multirow{2}{*}{$\m A_1$} & \multicolumn{2}{c|}{\multirow{2}{*}{$\omega_1\cdot T_{xy}^{(1)}$}}& \multirow{2}{*}{0} &\multicolumn{2}{c|}{\multirow{2}{*}{0}}		\\ 
	 & \multicolumn{2}{c|}{} 								&  	&	\multicolumn{2}{c|}{}	\\ \hline
	\multirow{2}{*}{$\vdots$} & \multicolumn{2}{c|}{\multirow{2}{*}{0}}  & \multirow{2}{*}{$\ddots$}	& \multicolumn{2}{c|}{\multirow{2}{*}{0}}\\ 
	&\multicolumn{2}{c|}{}&&\multicolumn{2}{c|}{}\\ 
	\hline
	\multirow{2}{*}{$\m A_l$} & \multicolumn{2}{c|}{\multirow{2}{*}{0}} & \multirow{2}{*}{0}& \multicolumn{2}{c|}{\multirow{2}{*}{$\omega_l\cdot T_{xy}^{(l)}$}} \\ 
	 & \multicolumn{2}{c|}{} &  &   \multicolumn{2}{c|}{} \\ \hline
\end{tabular}
\end{center}
\end{table}

\begin{lem}
\label{lem: key}
Let $p \in \m{C}_{qs}^{m,n,d,d}$ be a correlation on $\m{X},\m{Y}, \m{A}, \m{B}$, \induce d by a strategy \mbox{$(\ket{\psi} \in \m{H}_A \otimes \m{H}_B, \{\Pi_{A_x}^a\}_a, \{\Pi_{B_y}^b\}_b)$}. Suppose for some positive integer $l$, there exist partitions 
$\m{A} = \bigsqcup_{i=1}^{l}\m{A}_i$, $\m{B} = \bigsqcup_{i=1}^{l}\m{B}_i$, with $|\m{A}_i |= |\m{B}_i| = d_i$, and correlations $p_i \in \m{C}_{qs}^{m,n,d_i,d_i}$ on $\m{X},\m{Y}, \m{A}_i, \m{B}_i$ such that $p = \bigoplus_{i=1}^l \omega_i p_i$. Then there exist direct sum decompositions $\m H_A = \m H_A^\text{null} \oplus \bigoplus_i\m H_A^i, \m H_B = \m H_B^\text{null} \oplus \bigoplus_i\m H_B^i$ and strategies
\begin{equation}\label{eq: lem-key-strategy}
	\left(\frac{\ket{\psi_i}}{\norm{\ket{\psi_i}}} \in \m H_A^i\otimes \m H_B^i, \{\Pi_{A_x}^a|_{\m H_A^i}\}_{a\in \m{A}_i}, \{\Pi_{B_y}^b|_{\m H_B^i}\}_{b\in \m{B}_i}\right)
\end{equation}
 such that:
\begin{enumerate}[(i)]
    \item \label{item:lemma-key blocks-are-independent-of-questions}
    	Strategy \eqref{eq: lem-key-strategy} is well-defined, i.e.\ the restricted operators $\Pi^{a}_{\m A_x}|_{\m H_A^i}$ and $\Pi^{b}_{\m A_y}|_{\m H_B^i}$ are projections.
    \item \label{item:lemma-key block-weights-are-right}
	    $\|\ket{\psi_i}\|^2 = \omega_i$.
    
    \item \label{item:lemma-key strategy-works}
	    $p_i$ is \induce d by strategy \eqref{eq: lem-key-strategy}.
	\item
	\label{item:lemma-key 5}
	    For all $ x \in \m{X}, y \in \m{Y}, a \in \m{A}_i, b \in \m{B}_i$: $$\Pi_{A_x}^a|_{\m H_A^i} \ket{\psi_i} = \Pi_{A_x}^a \ket\psi, \,\,\,\, \Pi_{B_y}^B|_{\m H_B^i} \ket{\psi_i} = \Pi_{B_y}^b \ket\psi $$
\end{enumerate}
\end{lem}

\begin{proof}
For the remainder of the proof, when an operator acts only on one tensor factor we omit writing the identity on the other factors. 

Our first goal is to construct the subspaces $\m H_{A}^i, \m H_{B}^i$. We first study the action of the projectors corresponding to answers in $\m A_i$ and $\m B_i$ on the state $\ket{\psi}$. We will use these properties to define the states $\ket{\psi_i}$. Then from these, we will construct $\m H_A^i$ and $\m H_B^i$.

For $x\in \m{X}, y \in \m Y$, define $\Pi_{A_x}^{\m{A}_i} := \sum_{a \in \m{A}_i} \Pi_{A_x}^{a}$ and $\Pi_{B_y}^{\m{B}_i} := \sum_{b \in \m{B}_i} \Pi_{B_y}^{b}$  
. 
We will show that $\Pi_{A_x}^{\m{A}_i} \ket\psi = \Pi_{B_y}^{\m{B}_i} \ket\psi$ for all $i,x,y$.     
For any $i \in [l], x \in \mathcal{X}, y \in \mathcal{Y}$,
\begin{align}
\Pi_{A_x}^{\m{A}_i} \ket{\psi} 
	&= \left(\sum_{a \in \m{A}_i} \Pi_{A_x}^{a}\right) \otimes I \ket{\psi} \\
	\label{eq: 14-b}
	&= \left(\sum_{a \in \m{A}_i} \Pi_{A_x}^{a}\right) \otimes \left(\sum_{b \in \m{B}} \Pi_{B_y}^{b}\right) \ket{\psi} \\
	&= \left(\sum_{a \in \m{A}_i} \Pi_{A_x}^{a}\right) \otimes \left(\sum_{b \in \m{B}_i} \Pi_{B_y}^{b}\right) \ket{\psi} \\
	&= \Pi_{A_x}^{\m{A}_i} \otimes \Pi_{B_y}^{\m{B}_i} \ket{\psi}. \label{eq: 14}
\end{align}
The second equality follows from the fact that $\set{\Pi_{B_y}^{b}}$ forms a complete measurement. The third equality comes from the block structure of the correlation. More specifically, suppose that $a\in \m A_i$ but $b\nin \m B_i$. Then the block structure demands that $p(a,b| x,y) = 0$ for all $x,y$. So we conclude that
	$\norm{\Pi_{A_x}^{a} \otimes \Pi_{B_y}^{b} \ket \psi}^2= p(a,b|x,y)= 0.$
This forces the appropriate terms of the sum in Equation \eqref{eq: 14-b} to vanish.
 The same argument with the roles of $\m A$ and $\m B$ reversed gives 
 \begin{equation}\label{eq:15}
 	\Pi_{B_y}^{\m{B}_i} \ket{\psi} = \Pi_{A_x}^{\m{A}_i} \otimes \Pi_{B_y}^{\m{B}_i} \ket{\psi}.
 \end{equation}
 Combined with Equation \eqref{eq: 14}, this implies that, for any $i,x,y$, \begin{equation}
 \label{eq: action on state}
  \Pi_{A_x}^{\m{A}_i} \ket{\psi}= \Pi_{B_y}^{\m{B}_i} \ket{\psi}   
 \end{equation} 
 In particular, the action of $\Pi_{A_x}^{\m{A}_i}$ on $\ket\psi$ is the same for all $x$, and similarly for the $\m B$ operators. This lets us define $$\ket{\psi_i} := \Pi_{A_x}^{\m A_i} \ps,$$ where the choice of $x$ does not matter.


Now we compute the norm of $\ket{\psi_i}$. The block structure $p = \oplus_i \omega_i p_i$ of the correlation gives us that for any fixed $x$ and $y$,
\begin{align*}
	\omega_i 
	&= \sum_{a\in \m A_i, b \in \m B_i} p(a,b|x,y)
	\\&= \sum_{a\in \m A_i, b \in \m B_i} \expectation \psi {\Pi_{A_x}^{a} \otimes \Pi_{B_y}^{b}}
	\\&= \expectation \psi {\Pi_{A_x}^{\m{A}_i}\otimes \Pi_{B_y}^{\m{B}_i}}
	\\&= \norm{\ket{\psi_i}}^2.
\end{align*}
where the last line follows from Equation \eqref{eq: action on state}. 
This establishes condition \eqref{item:lemma-key block-weights-are-right}. 
Now let $\rho_{A}^i = \tr_B \proj{\psi_i} = \sum_j \lambda_j \proj j$, where $\lambda_j$ are the eigenvalues and $\ket{j}$ the eigenvectors of $\rho_{A}^i$. These are guaranteed to exist even if $\ket{\psi_i}$ is infinite-dimensional, because the existence of a Schmidt decomposition for any bipartite state holds also in infinite-dimensional Hilbert spaces. Notice that 
\begin{equation}
	\sum_j \lambda_j = \tr\rho_A^i = \norm{\ket{\psi_i}}^2 = \omega_i.
\end{equation}
We wish to compute the action of $\Pi_{A_x}^{\m A_i}$ on the eigenstates of $\rho_{A}^i$. We calculate
\begin{align*}
	\omega_i
	&= \expectation{\psi}{\Pi_{A_x}^{\m A_i} \otimes I} 
	\\&= \tr \Pi_{A_x}^{\m A_i}\rho_A^i
	\\&= \sum_j \lambda_j \tr \Pi_{A_x}^{\m A_i}\proj j
	\\&= \sum_j \lambda_j \norm{\Pi_{A_x}^{\m A_i}\ket j}^2.
\end{align*}
Since $\omega_i = \sum_j \lambda_j$, we must have $\norm{\Pi_{A_x}^{\m A_i}\ket j}^2 = 1$ for each $j$. In other words, $\Pi_{A_x}^{\m A_i}\ket j = \ket j$. 
This motivates us to define the space $\m H_A^i$ as the span of the nontrivial eigenvectors of $\rho_A^i$. Define also $P_i$ as the projection onto subspace $\m H_A^i$.

It follows from the definition of the $\ket{\psi_i}$ and the $\m H_A^i$ that
\begin{equation}
\label{eq: 18}
P_i \ket{\psi_j} = \delta_{ij}\ket{\psi_i}.
\end{equation}
Furthermore, notice that $\Pi_{A_x}^{\m A_i}P_i = P_i$. 
Thus the $\m H_A^i$ are suitable spaces for the new strategies to be defined on. In particular, the restricted operators $\Pi_{A_x}^{a}|_{\m H_A^i}$ are projectors. To see this, notice that they are orthogonal for distinct $a$ and that they sum to identity.

Let $\m H_A^\text{null}$ be the orthogonal complement of $\bigoplus_i \m H_A^i$ in $\m H_A$. 
Define $\m H_B^i$ and $\m H_B^\text{null}$ analogously. Clearly, $\bigoplus_i \m H_A^i$ and $\bigoplus_i \m H_B^i$ are topologically closed. This implies that $\m H_A = \m H_A^{\text{null}} \oplus \bigoplus_i \m H_A^i$ and $\m H_B = \m H_B^{\text{null}} \oplus \bigoplus_i \m H_B^i$.

Thus, we have established condition \eqref{item:lemma-key blocks-are-independent-of-questions} of the lemma.

It follows straightforwardly from the Definition of $\ket{\psi_i}$ and \eqref{eq: 18} that for $a\in \m A_i$, $\Pi_{A_x}^a|_{\m H_A^i} \ket{\psi_i} = \Pi_{A_x}^a \ket\psi$, and similarly for $\m B$. This establishes condition \eqref{item:lemma-key 5}. Finally, we show 
condition \eqref{item:lemma-key strategy-works}, that the strategies in each block \induce\ the appropriate correlations. We fix arbitrary $a \in \m A_i, b\in \m B_i, x \in \m X, y\in \m Y$, and calculate
\begin{align*}
    \frac{1}{\norm{\ket{\psi_i}}^2}
    \left\langle\psi_i\middle |{
    \Pi_{A_x}^a|_{\m H_A^i} \otimes \Pi_{B_y}^b|_{\m H_B^i}
    }\middle | \psi_i\right\rangle
    &=
    \frac1{\omega_i}
    \expectation{\psi}{
    \Pi_{A_x}^a \otimes \Pi_{B_y}^b
    }
    \\&= \frac1{\omega_i} p(a,b|x,y)
   \\& =p_i(a,b|x,y).
\end{align*}
In the above, the first quantity is the correlation \induce d by the strategy defined in Equation \eqref{eq: lem-key-strategy}, and the last quantity is the desired correlation $p_i$. Thus, we have shown condition \eqref{item:lemma-key strategy-works}.

\end{proof}

\section{The separating correlation}

\label{sec: separating correlation}

In this section, we describe the correlation $p^*$ that separates $\m{C}_q$ and $\m{C}_{qs}$. The correlation is on question sets $\m{X} = \{0,1,2,3\}$ and $\m{Y} = \{0,1,2,3,4\}$ and answer sets $\m{A} = \m{B} = \{0,1,2\}$. Hence, the smallest classes we separate are $\m{C}_{q}^{4,5,3,3}$ and $\m{C}_{qs}^{4,5,3,3}$. We define $p^*$ by describing the ideal infinite-dimensional strategy that \induce s it. In the following section, we will prove that no finite-dimensional strategy \induce s $p^*$.

Recall the definition of $\C^{\N}$ from section \ref{sec: preliminaries}. For each $m \geq 0$, we define two isometries $\veven m, \vodd m: \C^2 \to \C^\N$ as follows:
\begin{equation}
	\veven m\ket 0 = \ket {2m},
	\veven m\ket 1 = \ket {2m+1}
	\text{, and }
	\vodd m\ket 0 = \ket {2m+1},
	\vodd m\ket 1 = \ket {2m+2}.
\end{equation}

We use these isometries to define observables on $\C^\N$. By abuse of notation, for an isometry $V: \C^2 \rightarrow C^{\N}$ and an operator $O$ on $\C^2$, we write $V(O)$ to refer to the pushforward $VO V^\dagger$ of $O$ along $V$. For example, $\veven m (\sigma^z) = \ket{2m}\bra{2m} - \ket{2m+1}\bra{2m+1}$. For $O$ an operator with $+1,0,-1$ eigenvalues, we write $O^+$ for the projection onto the $+1$ eigenspace and $O^-$ for the projection onto the $-1$ eigenspace. One can check that with this notation $O = O^+ - O^-$.
We use the notation $\bigoplus A_i$ to denote the direct sum of observables $A_i$. We will make use of the $\alpha$-tilted Paulis $\saz,\sax$ from Definition \ref{def: ideal tilted chsh}. The following is the ideal strategy in detail. 

\begin{definition}[Ideal state and measurements for $p^* \in \m{C}_{qs}^{4,5,3,3}$]
\label{def: ideal correlation}
Fix $\alpha \in (0,1)$. The correlation $p^* \in \m{C}_{qs}^{4,5,3,3}$ is specified by the quantum strategy $(\ket{\Psi} \in \C^{\N} \otimes \C^{\N}, \{\Pi_{A_x}^a\}_a, \{\Pi_{B_y}^b\}_b\})$, where
   $\ket{\Psi} = \sqrt{1-\alpha^2} \sum_{i=0}^{\infty} \alpha^{i} \ket{ii}$, and the ideal measurements are described in Tables \ref{tab:alice-ideal} and \ref{tab:bob-ideal}.

\begin{table}[H]
    \caption{Alice's ideal measurements. The entry in cell $x,a$ is the projector $\Pi_{A_x}^a$.}
    \label{tab:alice-ideal}
    \centering
    \begin{tabular}{| c || Sc | Sc | Sc | Sc | Sc |}
	\hline
	\diagbox[width=2.5em]{$x$}{$a$} & 0 & 1 & 2 \\ \hhline{|====|}
	0 & $[\bigoplus_{m=0}^{\infty} \veven m(\sigma^z)]^+$ & $[\bigoplus_{m=0}^{\infty} \veven m (\sigma^z)]^-$ & 0 \\ \hline
	1 & $[\bigoplus_{m=0}^{\infty} \veven m(\sigma^x)]^+$ & $[\bigoplus_{m=0}^{\infty} \veven m (\sigma^x)]^-$ & 0 \\ \hline
	2 & $[\bigoplus_{m=0}^{\infty} \vodd m(\sigma^z)]^-$ & $[\bigoplus_{m=0}^{\infty} \vodd m (\sigma^z)]^+$ & $\proj 0$ \\ \hline
	3 & $[\bigoplus_{m=0}^{\infty} \vodd m(\sigma^x)]^-$ & $[\bigoplus_{m=0}^{\infty} \vodd m (\sigma^x)]^+$ & $\proj 0$ \\ 
	\hline
\end{tabular}
\end{table}

\begin{table}[H]
    \centering
    \caption{Bob's ideal measurements. The entry in cell $y,b$ is the projector $\Pi_{B_y}^b$.}
    \label{tab:bob-ideal}
    \begin{tabular}{| c || Sc | Sc | Sc | Sc | Sc |}
	\hline
	\diagbox[width=2.5em]{$y$}{$b$} & 0 & 1 & 2 \\ \hhline{|====|}
	0 & $[\bigoplus_{m=0}^{\infty} \veven m(\saz)]^+$ & $[\bigoplus_{m=0}^{\infty} \veven m (\saz)]^-$ & 0 \\ \hline
	1 & $[\bigoplus_{m=0}^{\infty} \veven m(\sax)]^+$ & $[\bigoplus_{m=0}^{\infty} \veven m (\sax)]^-$ & 0 \\ \hline
	2 & $[\bigoplus_{m=0}^{\infty} \vodd m(\saz)]^-$ & $[\bigoplus_{m=0}^{\infty} \vodd m (\saz)]^+$ & $\proj 0$ \\ \hline
	3 & $[\bigoplus_{m=0}^{\infty} \vodd m(\sax)]^-$ & $[\bigoplus_{m=0}^{\infty} \vodd m (\sax)]^+$ & $\proj 0$ \\ 
	\hline
	4 & $[\bigoplus_{m=0}^{\infty} \veven m(\sigma^z)]^+$ & $[\bigoplus_{m=0}^{\infty} \veven m (\sigma^z)]^-$ & 0 \\ \hline
\end{tabular}
\end{table}

\end{definition}

Intuitively, for questions $x,y \in \{0,1\}$, Alice and Bob decompose the space into a direct sum of $2\times 2$ blocks and perform the ideal tilted CHSH measurements for ration $\alpha$ on each block. For $x, y \in \{2,3\}$, they do the same, but with a block structure which is shifted forward by one standard basis element. Additionally, Bob has a fifth question on which he performs the same measurement as Alice performs on question $x=0$. 


The ideal state and measurements defining $p^*$ specify correlation tables $T_{xy}$ for all pairs of questions $x \in \{0,1,2,3\}$, $y \in \{0,1,2,3,4\}$. We explicitly report some of them, as we will later make use of the relations that these impose on the measurement projectors. For ease of notation let $C = \frac{1}{1-\alpha^2}$ in the tables below (note $C > 1$).
\begin{table}[H]
\caption{On the left, $T_{xy}$ for $x,y \in \{0,1\}$. The top-left $2 \times 2$ block contains ideal tilted CHSH correlations for questions $x,y$.}
\label{tab:txy1}
\begin{center}
\hfill
\begin{tabular}{| c || >{\centering}p{1cm} | >{\centering}p{1cm} | Sc | Sc | Sc |}
	\hline
	\diagbox[width=2.5em]{$a$}{$b$} & 0 & 1 & 2 \\ \hhline{|====|}
	0 & \multicolumn{2}{c|}{\multirow{2}{*}{CHSH$_{x,y}^{\alpha}$}}& 0 \\ 
	\cline{0-0}\cline{4-4}
	1 & \multicolumn{2}{c|}{} 								& 0  \\ \hline
	2 & 0 & 0 & 0\\
	\hline
\end{tabular}
\hfill
\begin{tabular}{| c || >{\centering}p{1cm} | >{\centering}p{1cm} | Sc | Sc | m{1em} |}
	\hline
	\diagbox[width=2.5em]{$a$}{$b$} & 1 & 0 & 2 \\ \hhline{|====|}
	1 & \multicolumn{2}{c|}{\multirow{2}{*}{$\frac{C-1}{C}\cdot\text{CHSH}_{\bar{x},\bar{y}}^{\alpha}$}}& 0 \\ 
	\cline{0-0}\cline{4-4}
	0 & \multicolumn{2}{c|}{} 								& 0  \\ \hline
	2 & 0 & 0 & $\frac1C$\\
	\hline
\end{tabular}
\hfill
\caption{On the right, $T_{xy}$ for $x,y \in \{2,3\}$. Let $\bar{x},\bar{y}$ be $x$, $y$ modulo $2$. The top-left $2 \times 2$ block contains the ideal tilted CHSH correlation table for questions $\bar{x},\bar{y}$, weighted by $\frac{C-1}{C}$ (notice that we have flipped the $0$ and $1$ labels in the rows and columns.)}
\label{tab:txy2}
\end{center}
\end{table}

\begin{table}[H]
\caption{On the left, $T_{xy}$ for $x = 0, y =4$}
\label{tab:txy3}
\begin{center}
\hfill
\begin{tabular}{| c || Sc | Sc | Sc | Sc | Sc |}
	\hline
	\diagbox[width=2.5em]{$a$}{$b$} & 0 & 1 & 2 \\ \hhline{|====|}
	0 & $ \frac{1}{C} \cdot \frac{1}{1-\alpha^4}$ & 0 & 0 \\ 
	\hline
	1 & 0 & $\frac{1}{C} \cdot\frac{\alpha^2}{1-\alpha^4}$ & 0  \\ \hline
	2 & 0 & 0 & 0\\
	\hline
\end{tabular}
\hfill
\begin{tabular}{| c || Sc | Sc | Sc | Sc | Sc |}
	\hline
	\diagbox[width=2.5em]{$a$}{$b$} & 0 & 1 & 2 \\ \hhline{|====|}
	0 & $ \frac{1}{C} \cdot (\frac{1}{1-\alpha^4}-1 )$ & 0 & 0 \\ 
	\hline
	1 & 0 & $\frac{1}{C} \cdot\frac{\alpha^2}{1-\alpha^4}$ & 0  \\ \hline
	2 & $\frac{1}{C}$ & 0 & 0\\
	\hline
\end{tabular}
\hfill
\caption{On the right, $T_{xy}$ for $x =2 ,y =4$}
\label{tab:txy4}
\end{center}
\end{table}

\section{Proof of separation}
\label{sec: proof}
In this section, we prove Theorem \ref{thm: main}. We start from a strategy that \induce s $p^*$: in \S \ref{sec: proj}, we prove properties of the state and the measurement operators, and in \S \ref{sec: schmidt}, we characterize the non-zero Schmidt coefficients, concluding that there must be infinitely many. 

\subsection{Characterizing the state and the projectors}
\label{sec: proj}



The following lemma establishes the existence of two local isometries which decompose any state achieving $p^*$ into two different ways (as anticipated in the ``cartoon proof'' of section \ref{sec: cartoon proof}).

\begin{lem}[Characterizing the state and projectors]
\label{lem: 2}
Let $(\ket{\psi} \in \m H_A \otimes \m H_B, \{\Pi_{A_x}^a\}, \{\Pi_{B_y}^b\})$ be a strategy inducing the ideal correlation $p^*$ from Definition \ref{def: ideal correlation}. Let $C = \frac{1}{1-\alpha^2}$. Then there exist two local isometries $\Phi$ and $\Phi'$ and (normalized) states $\ket{\mbox{aux}}$, $\ket{\mbox{aux}'}$ and $\ket{\mbox{aux}''}$ such that 
\begin{enumerate}[(i)]
    \item \label{item: lemma-2 first-isometry}
    \begin{itemize}
        \item $\Phi(\ket{\psi}) = \frac{1}{\sqrt{1+\alpha^2}}(\ket{00}+\alpha \ket{11}) \otimes \ket{\mbox{aux}}$
        \item $\Phi(\Pi_{A_0}^{0} \otimes I \ket{\psi}) = \frac{1}{\sqrt{1+\alpha^2}} \ket{00} \otimes \ket{\mbox{aux}}$
        \item $\Phi(\Pi_{A_0}^{1} \otimes I \ket{\psi}) = \frac{\alpha}{\sqrt{1+\alpha^2}} \ket{11} \otimes \ket{\mbox{aux}}$
    \end{itemize}
    \item \label{item: lemma-2 second-isometry}
    \begin{itemize}
	    \item $\Phi'(\ket{\psi}) = \frac{1}{\sqrt{C}}\ket{22} \otimes \ket{\mbox{aux}''} \oplus \sqrt{\frac{C-1}{C}}\frac{1}{\sqrt{1+\alpha^2}}(\ket{11}+\alpha \ket{00}) \otimes \ket{\mbox{aux}'}$
	    \item $\Phi'(\Pi_{A_2}^{0} \otimes I \ket{\psi}) = \sqrt{\frac{C-1}{C}} \frac{\alpha}{\sqrt{1+\alpha^2}}\ket{00} \otimes \ket{\mbox{aux}'}$
	    \item $\Phi'(\Pi_{A_2}^{1} \otimes I \ket{\psi}) = \sqrt{\frac{C-1}{C}} \frac{1}{\sqrt{1+\alpha^2}}\ket{11} \otimes \ket{\mbox{aux}'}$
	    \item $\Phi'(\Pi_{A_2}^{2} \otimes I \ket{\psi}) = \frac{1}{\sqrt{C}}\ket{22} \otimes \ket{\mbox{aux}''}$
    \end{itemize}
\end{enumerate}
\end{lem}

\begin{proof}
    \eqref{item: lemma-2 first-isometry}: Let $p'$ be the restriction of $p^*$ to questions $x,y \in \{0,1\}$. From Table \ref{tab:txy1}, we know that $p'$ is the ideal tilted CHSH correlation for ratio $\alpha$ (except that it has an extra answer ``$2$'' which has zero probability mass). Applying the block decomposition lemma (Lemma \ref{lem: key}) with $\omega_1 = 1$ and $\omega_2 = 0$, we have that there exist subspaces $\m{H}_{A}^1 \seq \m{H}_A$ and $\m{H}_B^1 \seq \m{H}_B$ such that the strategy $(\ket{\psi} \in \m{H}_A^1 \otimes \m{H}_B^1, \{\Pi_{A_x}^a|_{\m{H}_{A}^1}\}_{a\in \{0,1\}}, \{\Pi_{B_y}^b|_{\m{H}_{B}^1}\}_{b\in \{0,1\}})$ \induce s the ideal tilted CHSH correlation.
    
    By Lemma \ref{Bamps lemma}, the tilted CHSH correlation self-tests its ideal strategy, i.e.\ there exists a local isometry $\Phi_1 = \Phi_{1,A} \otimes \Phi_{1,B}$ with 
    $\Phi_{1,A}: \m{H}_A^1 \rightarrow \tilde{\m{H}}^{1}_{A} \otimes \tilde{\m{H}}^{1}_{A,aux}$ and 
    $\Phi_{1,B}: \m{H}_B^1 \rightarrow \tilde{\m{H}}^{1}_{B} \otimes \tilde{\m{H}}^{1}_{B,aux}$, and a (normalized) state $\ket{aux} \in \tilde{\m{H}}^{1}_{A,aux} \otimes \tilde{\m{H}}^{1}_{B,aux}$ such that $\Phi_1(\ket{\psi}) = \frac{1}{\sqrt{1+\alpha^2}}(\ket{00}+\alpha \ket{11}) \otimes \ket{aux}$. Moreover, by Lemma \ref{Bamps lemma}, it is also the case that 
    \begin{equation}
    	\Phi_{1}\left((\Pi_{A_0}^0|_{\m{H}_A^1} - \Pi_{A_0}^1|_{\m{H}_A^1}) \otimes I \ket{\psi}\right) = Z \otimes I \frac{1}{\sqrt{1+\alpha^2}}(\ket{00} + \alpha \ket{11})\otimes \ket{aux}. 
    \end{equation}
    Since $(I+Z)/2 = \proj{0}$ and $(I-Z)/2 = \proj{1}$, we deduce by linearity that 
    \begin{equation*}
    	\Phi_1\left(\Pi_{A_0}^0|_{\m{H}_A^1} \otimes I \ket{\psi}\right) = \frac{1}{\sqrt{1+\alpha^2}}\ket{00} \otimes \ket{aux}\text{ and }\Phi_1\left(\Pi_{A_0}^1|_{\m{H}_A^1} \otimes I \ket{\psi}\right) = \frac{\alpha}{\sqrt{1+\alpha^2}} \ket{11} \otimes \ket{aux}.
    \end{equation*}
    Letting $\Phi$ be any isometric extension of $\Phi_1$ to $\m{H}_A \otimes \m{H}_B$ and applying condition \eqref{item:lemma-key 5} of Lemma \ref{lem: key} gives \eqref{item: lemma-2 first-isometry}.
    
    \vspace{1.5mm}
    
\noindent \eqref{item: lemma-2 second-isometry}:  Let $p''$ be the restriction of $p^*$ to questions $x,y \in \{2,3\}$. Then from table \ref{tab:txy2} we have that $p'' = \omega_1 p_1 \oplus \omega_2 p_2$ where $p_1$ is the ideal tilted CHSH correlation (for ratio $\alpha$) and $p_2$ is the correlation in which answer $(2,2)$ has probability $1$ on all question pairs, and $\omega_1 = \frac{C-1}{C}$, $\omega_2 = \frac{1}{C}$.

	By Lemma \ref{lem: key}, there exist subspaces $\m H_A^\text{null}, \m H_B^\text{null}$, $\m{H}_A^1, \m{H}_A^2$, $\m{H}_B^1, \m{H}_B^2$ with $\m H_A = \m H_A^\text{null} \oplus \m{H}_A^1 \oplus \m{H}_A^2$ and $\m H_B = \m H_B^\text{null} \oplus \m{H}_B^1 \oplus \m{H}_B^2$,
    and strategies $S_1$ and $S_2$ with
    \begin{align*}
    S_1 &= \left(\frac{\ket{\psi_1}}{\|\ket{\psi}\|} \in \m{H}_A^1 \otimes \m{H}_B^1, \{\Pi_{A_x}^a|_{\m{H}_A^1}\}_{a\in \{0,1\}}, \{\Pi_{B_y}^b|_{\m{H}_B^1}\}_{b\in \{0,1\}}\right), \\
    S_2 &= \left(\frac{\ket{\psi_2}}{\|\ket{\psi_2}\|} \in \m{H}_A^2 \otimes \m{H}_B^2, \{\Pi_{A_x}^2|_{\m{H}_A^2}\}, \{\Pi_{B_y}^2|_{\m{H}_B^2}\}\right)
    \end{align*}
	such that $\|\ket{\psi_1} \|^2 = \frac{C-1}{C}$, $\|\ket{\psi_2} \|^2 = \frac{1}{C}$ and $\ket{\psi} = \ket{\psi_1} + \ket{\psi_2}$. Moreover, $S_1$ \induce s the ideal tilted CHSH correlation for ratio $\alpha$ (with the roles of the $0$ and $1$ answers flipped --- see Table \ref{tab:txy2}). As in the proof of \eqref{item: lemma-2 first-isometry}, we can apply Lemma \ref{Bamps lemma} to obtain
	local isometries $\Phi_1 = \Phi_{1,A} \otimes \Phi_{1,B}$ with 
		$\Phi_{1,A}: \m{H}_A^1 \rightarrow \tilde{\m{H}}^{1}_{A} \otimes \tilde{\m{H}}^{1}_{A,aux}$ and 
    $\Phi_{1,B}: \m{H}_B^1 \rightarrow \tilde{\m{H}}^{1}_{B} \otimes \tilde{\m{H}}^{1}_{B,aux}$, and a (normalized) state $\ket{aux'} \in \tilde{\m{H}}^{1}_{A,aux} \otimes \tilde{\m{H}}^{1}_{B,aux}$
	such that
\begin{enumerate}[(a)]
    \item $\Phi_1(\ket{\psi_1}) = \sqrt{\frac{C-1}{C}}\frac{1}{\sqrt{1+\alpha^2}}(\ket{11}+\alpha \ket{00}) \otimes \ket{aux'}$, (we have flipped the zero and one basis elements for later convenience)
    \item $\Phi_1(\Pi_{A_2}^1|_{\m{H}_A^1} \otimes I \ket{\psi_1}) = \sqrt{\frac{C-1}{C}} \frac{1}{\sqrt{1+\alpha^2}}\ket{11} \otimes \ket{aux'}$, and
    \item $\Phi_1(\Pi_{A_2}^0|_{\m{H}_A^1}  \otimes I \ket{\psi_1}) = \sqrt{\frac{C-1}{C}} \frac{\alpha}{\sqrt{1+\alpha^2}}\ket{00} \otimes \ket{aux'}$.
\end{enumerate}
where (b) and (c) are obtained similarly as in part \eqref{item: lemma-2 first-isometry} of this proof.

Now, let $\Phi_2 = \Phi_{2,A} \otimes \Phi_{2,B}$, with $\Phi_{2,A}: \m{H}_A^2 \rightarrow \tilde{\m{H}}^{2}_{A} \otimes \tilde{\m{H}}^{2}_{A,aux}$ and 
    $\Phi_{2,B}: \m{H}_B^2 \rightarrow \tilde{\m{H}}^{2}_{B} \otimes \tilde{\m{H}}^{2}_{B,aux}$ be a local isometry, and $\ket{aux''} \in \tilde{\m{H}}^{2}_{A,aux} \otimes \tilde{\m{H}}^{2}_{B,aux}$ a (normalized) state such that

\begin{enumerate}
\item[(d)] $\Phi_2(\ket{\psi_2}) = \frac{1}{\sqrt{C}}\ket{22} \otimes \ket{aux''}$.
\end{enumerate}
Such $\Phi_2$ and $\ket{aux''}$ trivially exist. 

Define 
\begin{itemize}
    \item $\Phi'_{A}: \m{H}_A^{1} \oplus  \m{H}_A^{2} \rightarrow (\tilde{\m{H}}^{(1)}_{A} \otimes \tilde{\m{H}}^{(1)}_{A,aux}) \oplus (\tilde{\m{H}}^{(2)}_{A} \otimes \tilde{\m{H}}^{(2)}_{A, aux})$ as $\Phi'_{A} = \Phi_{1,A} \oplus \Phi_{2,A}$
    \item $\Phi'_{B}: \m{H}_B^{1} \oplus  \m{H}_B^{2} \rightarrow (\tilde{\m{H}}^{(1)}_{B} \otimes \tilde{\m{H}}^{(1)}_{B,aux}) \oplus (\tilde{\m{H}}^{(2)}_{B} \otimes \tilde{\m{H}}^{(2)}_{B, aux})$ as $\Phi'_{B} = \Phi_{1,B} \oplus \Phi_{2,B}$
\end{itemize}
Let $\Phi_{A}''$ be any isometric extension of $\Phi_A'$ to $\m H_{A}$, and let $\Phi_{B}''$ be any isometric extension of $\Phi_B'$ to $\m H_{B}'$. Let $\Phi' = \Phi''_A \otimes \Phi_B''$. Then (a), (b), (c) and (d), together with condition \eqref{item:lemma-key 5} of Lemma \ref{lem: key}, imply that $\Phi'$ satisfies condition \eqref{item: lemma-2 second-isometry} of Lemma \ref{lem: 2}, as desired.

\end{proof}

We also need the following properties, obtained using the $y=4$ question on Bob's side.
\begin{lem}
\label{lem: 3}
Let $(\ket{\psi}, \{\Pi_{A_x}^a\}, \{\Pi_{B_y}^b\})$ be a strategy inducing $p^*$.
The following properties hold:
\begin{enumerate}[(i)]
    \item \label{item: lemma 9-1} 
	    $\Pi_{A_0}^{0} \ket{\psi} = \Pi_{B_4}^{0} \ket{\psi} =  (\Pi_{A_2}^2 + \Pi_{A_2}^0) \ket{\psi} $
    \item \label{item: lemma 9-2} 
	    $\Pi_{A_0}^{1} \ket{\psi} = \Pi_{B_4}^{1} \ket{\psi} =  \Pi_{A_2}^1 \ket{\psi}$
    \item \label{item: lemma 9-3} 
	    $\ket\psi = \Pi_{A_0}^{0} \otimes \Pi_{B_4}^0 \ket{\psi} + \Pi_{A_0}^{1}\otimes \Pi_{B_4}^{1} \ket{\psi}$
\end{enumerate}
\end{lem}
\begin{proof}
    From correlation table \ref{tab:txy3}, we read out that $\bra{\psi} \Pi_{A_0}^0 \Pi_{B_4}^0 \ket{\psi} = \| \Pi_{A_0}^{0} \ket{\psi}\|^{2} = \| \Pi_{B_4}^{0} \ket{\psi}\|^{2}$. By the Cauchy-Schwarz inequality, this implies that $\Pi_{A_0}^{0} \ket{\psi} = \Pi_{B_4}^{0} \ket{\psi}$. Similarly, from correlation table \ref{tab:txy4}, we deduce $(\Pi_{A_2}^{2} + \Pi_{A_2}^{0}) \ket{\psi} = \Pi_{B_4}^{0} \ket{\psi}$, which yields \eqref{item: lemma 9-1}. We derive \eqref{item: lemma 9-2} analogously. Item \eqref{item: lemma 9-3} follows from combining the previous two items with the equality $(\Pi_{A_0}^0 + \Pi_{A_0}^1)\ket\psi = \ket\psi$.
\end{proof}

\subsection{Characterizing the Schmidt coefficients}
\label{sec: schmidt}

From now onwards, let $(\ket{\psi}, \{\Pi_{A_x}^a\}, \{\Pi_{B_y}^b\})$ be any strategy inducing $p^*$. 
In the previous subsection, we gave a partial characterization of the operators and state
 . In this subsection, we make use of these properties to show that $\ket\psi$ has infinitely many non-zero Schmidt coefficients. 
 For a bipartite state $\ket{\phi}_{AB}$, we denote by $\Schmidt\left(\ket\phi_{AB}\right)$ the multiset\footnote{Here by multiset we mean a set with multiplicity, sometimes called an unordered list. For example, the multiset of Schmidt coefficients of the EPR pair is $(\frac1{\sqrt2}, \frac1{\sqrt2})$.} of non-zero Schmidt coefficients of $\ket{\phi}_{AB}$. Recall that the Schmidt coefficients $\set{\lambda_i}$ are the unique nonnegative real numbers so that $\ket\phi_{AB} = \sum_i \lambda_i \ket i_A\otimes \ket i_B$ for some bases of the $A$ and $B$ registers. Any such pair of bases is called a pair of \emph{Schmidt bases with respect to $\ket\phi$.}
 Usually the tensor product decomposition of the Hilbert space will be clear, in which case we'll simply write $\Schmidt(\ket\phi)$ without the subscripts.
 We will use the following basic fact about Schmidt coefficients; we provide a proof for completeness.

 \begin{lem}\label{lem:schmidt-sum}
 	Let $\ket\psi, \ket\phi,\ket\eta$ be states on $\m H_A \otimes \m H_B$ with $\ket\psi = \ket\phi + \ket \eta$. Define reduced densities
 	\begin{equation}
 	    \rho_A = \tr_B \proj \psi, \sigma_A = \tr_B \proj \phi, \tau_A = \tr_B \proj \eta
 	\end{equation}
 	on $\m H_A$.
 	Define $\rho_B, \sigma_B, \tau_B$ similarly. Suppose that $\ket\phi$ and $\ket\eta$ are ``orthogonal on both subsystems'' in the sense that $\sigma_A\tau_A = 0 = \sigma_B\tau_B$. 
 	 Then $\Schmidt(\ket\psi) = \Schmidt(\ket \phi) \sqcup \Schmidt(\ket\eta)$, where $\sqcup$ denotes disjoint union.
 \end{lem}
 \begin{proof}
 	A Schmidt basis for $\m H_A$ with respect to $\ket\psi$ is the same as an eigenbasis for the reduced density operator $\tr_B \proj\psi$. Using the orthogonality of $\ket\phi$ and $\ket\eta$, one can check that the three densities $\rho_A, \sigma_A, \tau_A$ commute. Therefore, the densities have a common eigenbasis. This is also a common Schmidt basis. After repeating the argument to find a common Schmidt basis on $\m H_B$, we can write the states as
 	\begin{equation}
 		\ket\psi = \sum_i \lambda_i\ket {ii},
 		\ket\phi = \sum_i a_i\ket {ii},\text{ and }
 		\ket\eta = \sum_i b_i\ket {ii},
 	\end{equation}
 	with $a_i + b_i = \lambda_i$. By the orthogonality of $\ket\eta$ and $\ket\phi$, we have $a_ib_i = 0$ for each $i$. This implies that for each $i$, exactly one of the following two equalities holds: $\lambda_i = a_i$ or $\lambda_i = b_i$. This yields the lemma.
 \end{proof}

\begin{lem}
\label{lem: 4}
    Let $\Phi$, $\Phi'$ and $\ket{aux}$, $\ket{aux'}$, $\ket{aux''}$
be the local isometries and auxiliary states from Lemma \ref{lem: 2}.
Let $S = \Schmidt\left(\ps\right)$, and let $S_2 = \Schmidt\left(\frac{1}{\sqrt{C}}\ket{22} \otimes \ket{aux''}\right)$. Then there exists a partition $S = S_0 \sqcup S_1$ such that:
\begin{itemize}
    \item $S_0 =  \Schmidt\left( \frac{1}{\sqrt{1+\alpha^2}}\ket{00} \otimes \ket{aux}\right) = S_2 \sqcup \Schmidt\left( \sqrt{\frac{C-1}{C}} \frac{\alpha}{\sqrt{1+\alpha^2}}\ket{00} \otimes \ket{aux'} \right)$
    \item $S_1 =  \Schmidt\left( \frac{\alpha}{\sqrt{1+\alpha^2}}\ket{11} \otimes \ket{aux}\right) = \Schmidt\left(\sqrt{\frac{C-1}{C}} \frac{1}{\sqrt{1+\alpha^2}}\ket{11} \otimes \ket{aux'} \right)$
\end{itemize}
\end{lem}
\noindent
Notice that these two equalities give us two different correspondences between the Schmidt coefficients of $\ket{aux}$ and $\ket{aux'}$, where one involves multiplying by $\alpha$ and the other involves dividing by $\alpha$.

\begin{proof}
	Recall from Lemma \ref{lem: 3} that
	 $\ket{\psi} = \Pi_{A_0}^{0} \otimes \Pi_{B_4}^0 \ket{\psi} + \Pi_{A_0}^{1}\otimes \Pi_{B_4}^{1} \ket{\psi}$
	 . 
	 We deduce by Lemma \ref{lem:schmidt-sum} that $S$ can be partitioned into two sets $S_0$ and $S_1$, where 
   \begin{equation}
       S_0 =  \Schmidt\left( \Pi_{A_{0}}^{0} \ket{\psi}\right) \,\,\,\mbox{and    }
       S_1 =  \Schmidt\left( \Pi_{A_{0}}^{1} \ket{\psi}\right). \label{eq: 22}
   \end{equation}

\noindent
Since local isometries preserve Schmidt coefficients,  $\Phi(\ket{\psi})$, $\Phi'(\ket{\psi})$ and $\ket{\psi}$ have the same set of Schmidt coefficients $S$. Moreover, Lemma \ref{lem: 2} gives 
\begin{equation}
	\Phi(\Pi_{A_{0}}^{0} \ket{\psi}) = \frac{1}{\sqrt{1+\alpha^2}} \ket{00} \otimes \ket{aux}
	\text{ and }
	\Phi(\Pi_{A_{0}}^{1} \ket{\psi}) = \frac{\alpha}{\sqrt{1+\alpha^2}} \ket{11} \otimes \ket{aux}. 
\end{equation}
By direct substitution, 
\begin{equation}
	\label{eq: 23}
       S_0 =  \Schmidt\left( \frac{1}{\sqrt{1+\alpha^2}} \ket{00} \otimes \ket{aux}\right) 
       \text{ and }  
       S_1 =  \Schmidt\left( \frac{\alpha}{\sqrt{1+\alpha^2}} \ket{11} \otimes \ket{aux}\right). 
   \end{equation}
\noindent
By Lemma \ref{lem: 3}, we also have $\Pi_{A_0}^{0} \ket{\psi} = (\Pi_{A_2}^{2} + \Pi_{A_2}^{0}) \ket{\psi}$ and $\Pi_{A_0}^{1}\ket\psi = \Pi_{A_2}^{1} \ket{\psi}$. Moreover, from Lemma \ref{lem: 2}, we also have 
$\Phi'((\Pi_{A_2}^{2} + \Pi_{A_2}^{0}) \ket{\psi} )= \frac{1}{\sqrt{C}}\ket{22} \otimes \ket{aux''} + \sqrt{\frac{C-1}{C}} \frac{\alpha}{\sqrt{1+\alpha^2}}\ket{00} \otimes \ket{aux'}$ and $\Phi'(\Pi_{A_2}^{1} \ket{\psi}) = \sqrt{\frac{C-1}{C}} \frac{1}{\sqrt{1+\alpha^2}}\ket{11} \otimes \ket{aux'}$. Then this implies 

\begin{align}
       S_0 &=  \Schmidt\left( \frac{1}{\sqrt{C}}\ket{22} \otimes \ket{aux''} + \sqrt{\frac{C-1}{C}} \frac{\alpha}{\sqrt{1+\alpha^2}}\ket{00} \otimes \ket{aux'}\right) \nonumber\\
       &= S_2 \sqcup \Schmidt\left(\sqrt{\frac{C-1}{C}} \frac{\alpha}{\sqrt{1+\alpha^2}}\ket{00} \otimes \ket{aux'}\right),   \,\,\mbox{and    } \label{eq: 25} \\
       S_1 &=  \Schmidt\left( \sqrt{\frac{C-1}{C}} \frac{1}{\sqrt{1+\alpha^2}}\ket{11} \otimes \ket{aux'}\right). \label{eq: 26}
   \end{align}

Putting together Equations \eqref{eq: 22} through \eqref{eq: 26} gives the statement of the Lemma.
\end{proof}

\begin{theorem}
	Let $p^*$ be the ideal correlation introduced in Definition \ref{def: ideal correlation}.
 Let $(\ket{\psi}, \{\Pi_{A_x}^a\}, \{\Pi_{B_y}^b\})$ be any strategy inducing $p^*$. Then $\ps$ has infinitely many non-zero Schmidt coefficients.
\end{theorem}
\begin{proof}
 Let $\ket{aux},\ket{aux'}$, $S_0$, $S_1$ and $S_2$ be as in Lemma \ref{lem: 4}. Recall from Lemma \ref{lem: 4} that
\begin{equation}
	S_0 =  \Schmidt\left( \frac{1}{\sqrt{1+\alpha^2}}\ket{00} \otimes \ket{aux}\right)\text{ and }S_1 = \Schmidt\left(\frac{\alpha}{\sqrt{1+\alpha^2}}\ket{11} \otimes \ket{aux}\right). 
\end{equation}
 Then we can rewrite these sets as
 \begin{equation}
 	S_0 = \left\{\frac{1}{\sqrt{1+\alpha^2}} \lambda: \lambda \in \Schmidt\left(\ket{aux}\right)\right\}
 	\text{ and }
 	S_1 = \left\{\frac{1 }{\sqrt{1+\alpha^2}}\alpha\lambda : \,\lambda \in \Schmidt\left(\ket{aux}\right)\right\}
 \end{equation}
 Notice that there is a bijection $f:S_0\to S_1$ such that $f(\lambda) = \alpha \lambda$. Again from Lemma \ref{lem: 4} we have 
 \begin{equation*}
 S_0 = S_2 \sqcup \Schmidt\left( \sqrt{\frac{C-1}{C}} \frac{\alpha}{\sqrt{1+\alpha^2}}\ket{00} \otimes \ket{aux'} \right)
 \text{ and }
 S_1 = \Schmidt\left(\sqrt{\frac{C-1}{C}} \frac{1}{\sqrt{1+\alpha^2}}\ket{11} \otimes \ket{aux'} \right). 
 \end{equation*}

 Then we can rewrite $S_0\setminus S_2$ and $S_1$ as
 \begin{equation*}
	 S_0 \setminus S_2 = \left\{\sqrt{\frac{C-1}{C}} \frac{\alpha}{\sqrt{1+\alpha^2}} \lambda : \lambda \in \Schmidt\left(\ket{aux'}\right)\right\}
	 \text{ and }
	 S_1 = \left\{\sqrt{\frac{C-1}{C}} \frac{1}{\sqrt{1+\alpha^2}} \lambda: \lambda \in \Schmidt\left(\ket{aux'}\right)\right\}. 
 \end{equation*}
 Notice that there is a bijection $g: S_1\to S_0 \setminus S_2$ such that $g(\lambda) = \alpha\lambda$.
 
Composing the maps $f$ and $g$ yields a bijection between $S_0$ and $S_0 \setminus S_2$. Since $S_2$ is nonempty, this implies that $S_0$ must be infinite. 
\end{proof}
One can extend this proof a bit farther. Repeated applications of the map $f\circ g$ show that $S_0$ has an infinite descending sequence of the form $(\lambda, \alpha^2\lambda, \alpha^4\lambda,\ldots)$. One more application $f$ then shows that $S$ has an infinite sequence $(\lambda, \alpha\lambda, \alpha^2\lambda, \alpha^3\lambda,\ldots)$ This can be used to obtain some quantitative bounds on the dimension required to induce a correlation close to the ideal one. We do not prove this quantitative bound because much more useful bounds already exist for correlations witnessing the separation $\m C_{qs}\neq \m C_{qa}$. 

\section{Conclusion}
\label{sec: conclusion}
In this work, we answered affirmatively the long-standing question of whether $\m C_q$ is strictly contained in $\m C_{qs}$. We explicitly provided a correlation, on question sets of size $4$ and $5$ and answer sets of size $3$, that can be \induce d by an infinite-dimensional quantum strategy but not a finite-dimensional one. The construction of the correlation and the proof of separation are inspired by self-techniques which allow to characterize the structure of quantum strategies achieving correlations that possess a direct sum form. 

A promising avenue for further investigation is the possibility of applying such techniques to study the $I_{3322}$ inequality \cite{I3322Froissart}, which is also suspected to witness separation of $\m C_q$ and $\m C_{qs}$ (but with slightly smaller question and answer sets than ours). The $(3,3,2,2)$ scenario is the simplest one that is suspected to separate $\m C_q$ and $\m C_{qs}$. Numerical evidence \cite{I3322Pal} suggests that the ideal measurements achieving the conjectured maximal violation have a block-diagonal form. This suggests that the study of this inequality is potentially amenable to ideas and techniques from our work.

\section*{Acknowledgements}
The authors thank Laura Man\v cinska, William Slofstra and Thomas Vidick for useful discussions. The authors also thank Thomas Vidick for valuable comments on an earlier version of this paper.  
A.C. is supported by the Kortschak Scholars program and AFOSR YIP award number FA9550-16-1-0495. J.S. is supported by NSF CAREER Grant CCF-1553477 and the Mellon Mays Undergraduate Fellowship.

\bibliographystyle{alpha}
\bibliography{references}

\newpage

\end{document}